\documentclass[letterpaper, 10 pt]{report}  
\usepackage{graphicx}
\usepackage{algorithm2e}
\usepackage{color}
\usepackage{wrapfig}

\usepackage{comment}
\usepackage{amsmath,amssymb,amsthm}
\usepackage{mathtools}
\usepackage{flushend}

\usepackage{tikz}
\usepackage{amssymb}
\usepackage{subcaption}
 \usepackage{paralist}
 \usepackage{acronym}

\usepackage{mathtools}
\DeclarePairedDelimiter{\ceil}{\lceil}{\rceil}

\newtheorem{lemma}{Lemma}

\newtheorem{definition}{Definition}

\newtheorem{theorem}{Theorem}

\newif\ifuseboldmathops
\useboldmathopstrue 

\ifuseboldmathops
	\newcommand{\reals}{\mathbf{R}}

\else
	\newcommand{\reals}{\mathbb{R}}

\fi

\newcommand{\norm}[1]{\lVert #1 \rVert}

\newcommand{\vmu}{\mathbf{\mu}}

\newcommand{\rvw}{\mathbf{w}}
\newcommand{\bbE}{\mathbb{E}}
\newcommand{\abs}[1]{\lvert#1\rvert}

\newcommand{\bbR}{\mathbb{R}}

\newcommand{\origin}{\mathbf{0}}

\newcommand{\goal}{\mathsf{goal}}
\acrodef{lti}[LTI]{linear time invariant}
\acrodef{mras}[MRAS]{Model Reference Adaptive Search}
\acrodef{mdp}[MDP]{Markov decision process}
\acrodef{saop}[SAOP]{Sampling-based Approximate-Optimal Planning}
\acrodef{mdp}[MDP]{Markov decision process}

\acrodef{hjb}[HJB]{Hamilton-Jacobi-Bellman}
\acrodef{kl}[KL]{Kullback–Leibler} 
\acrodef{rbf}[RBF]{Radial basis function}
\acrodef{mpc}[MPC]{Model Predictive Control}
\acrodef{asaop}[ASAOP]{Adaptive
Search-based Approximate-Optimal Planning}
 \newtheorem{assumption}{Assumption}
\usepackage{siunitx}
\usepackage{subcaption}
\usepackage{graphicx}
\usepackage{float}

\title{\LARGE \bf Importance sampling-based approximate optimal
  planning and control}

\author{Jie Fu$^{1}$ 
\thanks{$^{1}$Jie Fu
    is with Robotics Engineering Program, Department of Electrical and
    Computer Engineering, Worcester Polytechnic Institute, 01609,
    Worcester, MA, US.  {\tt\small jfu2@wpi.edu}
    }
}

\begin{document}
\maketitle

\thispagestyle{empty}
\pagestyle{empty}
\begin{abstract}In this paper, we propose a sampling-based planning and optimal control method of nonlinear systems under non-differentiable constraints. Motivated by developing scalable planning algorithms, we consider the optimal motion plan to be a feedback controller that can be approximated by a weighted sum of given bases. Given this approximate optimal control formulation, our main contribution is to introduce importance sampling, specifically, model-reference adaptive search algorithm, to iteratively compute the optimal weight parameters, i.e., the weights corresponding to the optimal policy function approximation given chosen bases. The key idea is to perform the search by iteratively estimating a parametrized distribution which converges to a Dirac's Delta that infinitely peaks on the global optimal weights. Then, using this direct policy search, we incorporated trajectory-based verification to ensure that, for a class of nonlinear systems, the obtained policy is not only optimal but robust to bounded disturbances.  The correctness and efficiency of the methods are demonstrated through numerical experiments including linear systems with a nonlinear cost function and motion planning for a Dubins car.
\end{abstract}

\section{Introduction}

This paper presents an importance sampling based approximate optimal
planning and control algorithm.  Optimal motion planning in
deterministic and continuous systems is computationally NP-complete
\cite{lavalle2006planning} except for linear time invariant
systems. For nonlinear systems, there is a vast literature on
approximate solutions and algorithms. In optimal planning, the common
approximation scheme is discretization-based.  By discretizing the
state and input spaces, optimal planning is performed by solving the
shortest path problem in the discrete transition systems obtained from
abstracting the continuous dynamics, using heuristic-based search or
dynamic programming.
Comparing to discretization-based methods, \emph{sampling-based graph
  search}, includes Probabilistic RoadMap (PRM)
\cite{kavraki1996probabilistic}, RRT \cite{Lavalle98rrt}, RRT*
\cite{karaman2011sampling}, are more applicable for high-dimensional
systems. While RRT has no guarantee on the optimality of the path
\cite{karaman2011sampling}, RRT* compute an optimal path
asymptotically provided the cost functional is Lipschitz
continuous. However, such Lipschitz conditions may not be satisfied
for some cost functions under specific performance consideration.



The key idea in the proposed sampling-based planning method builds on
a unification of importance sampling and approximate optimal control
\cite{bertsekas2011approximate,bertsekas2011dynamic}.  In approximate
optimal control, the objective is to approximate both the value
function, i.e., optimal cost-to-go, and the optimal feedback policy
function by weighted sums of \emph{known} basis functions.  As a
consequence, the search space is changed from infinite trajectory
space or policy space to a continuous space of weight vectors, given
that each weight vector corresponds to a unique feedback
controller. 
Instead of solving the approximate optimal control through training
actor and critic neural networks (NNs) using trajectory data
\cite{abu2005nearly,bertsekas1996neuro}, we propose a sampling-based
method for sampling the weight vectors for a policy function
approximation and searching for the optimal one. This method employs
\ac{mras} \cite{hu2007model}, a probabilistic complete global
optimization algorithm, for searching the optimal weight vector that
parametrizes the approximate optimal feedback policy. The fundamental
idea is to treat the weight vector as a random variable over a
parameterized distribution and the optimal weight vector corresponds
to a Dirac's Delta function which is the target distribution.  The
\ac{mras} algorithm iteratively estimates the parameter that possesses
the minimum Kullback-Leibler divergence with respect to an
intermediate reference model, which assigns a higher probability mass
on a set of weights of controllers with improved performance over the
previous iteration. At the meantime, a set of sampled weight vectors
are generated using the parameterized distribution and the performance
of their corresponding policies are evaluated via simulation-based
policy evaluation. Under mild conditions, the parameterized
distribution converges, with probability one, to the target
distribution that concentrates on the optimal weight vector with
respect to given basis functions.

\ac{mras} resembles another adaptive search algorithm called
cross-entropy(CE) method and provides faster and stronger convergence
guarantee for being less sensitive to input parameters
\cite{hu2007model,homem2007study}.  Previously, CE algorithm has been
introduced for motion planning
\cite{kobilarov2012cross,livingston2015cross} based on sampling in the
trajectory space. The center idea is to construct a probability
distribution over the set of feasible paths and to perform the search
for an optimal trajectory using CE. The parameters to be estimated is
either a sequence of motion primitives or a set of via-points for
interpolation-based trajectory planning.  Differ to these methods,
ours is the first to integrate importance sampling to estimate
parameterization of the optimal policy function approximation for
continuous nonlinear systems. Since the algorithm performs direct
policy search, we are able to enforce robustness and stability
conditions to ensure the computed policy is both robust and
approximate optimal, provided these conditions can be evaluated
efficiently.

To conclude, the contributions of this paper are the following: First,
we introduce a planning algorithm by a novel integration of model
reference adaptive search and approximate optimal control.  Second,
based on contraction theory, we introduce a modification to the
planning method to directly generate stabilizing and robust feedback
controllers in the presence of bounded disturbances.  Last but not the
least, through illustrative examples, we demonstrate the effectiveness
and efficiency of the proposed methods and share our view on
interesting future research along this direction.

\section{Problem formulation}

Notation: The inner product between two vectors $w,v\in \reals^n$ is
denoted $w^\intercal v$ or $ \langle w,v \rangle$.  Given a positive
semi-definite matrix $P$, the $P$-norm of a vector is denoted
$\norm{x}_P = \sqrt{x^\intercal P x}$. We denote $\norm{x}$ for $P$
being the identity matrix. $I_{\{E\}}$ is the indicator function,
i.e., $I_{E}=1$ if event $E$ holds, otherwise $0$.  For a real
$\alpha \in \bbR$, $\lceil \alpha\rceil$ is the smallest integer that
is greater than $\alpha$.
\subsection{System model}
We consider continuous-time nonlinear systems of the form

\begin{align}
\label{eq:sys}
\begin{split}
\Sigma: \quad & \dot x(t) = f(x(t),u(t)),\\
  &x(t)\in X, u(t)\in U.
\end{split}
\end{align}
where $x \in X$ is the state, $u \in U $ is the control input,
$x_0\in X$ is the initial state, and $f(x,u)$ is a vector field. We
assume that $X$ and $U$ are compact.  A feedback controller
$u:X\rightarrow U$ takes the current state and outputs a control
input. 

The objective is to find a feedback controller $u^\ast $ that
minimizes a  finite-horizon cost function for a nonlinear system
\begin{align}
\begin{split}
\label{eq:originproblem}
\min_{u} J( x_0, u)& =  \int_{0}^{T} \ell ( x(t), u(t))dt
+ g(x(T),u(T))\\
\mbox{subject to: }& \dot x(t) = f(x(t),u(t)),\\
  &x(t)\in X, \; u(t)\in U, \; x(0)=x_0.
\end{split}
\end{align}
where $T$ is the stopping time, $\ell: X\times U \rightarrow \bbR^+ $
defines the running cost when the state trajectory traverses through
$x$ and the control input $u$ is applied and $g: X\rightarrow\bbR^+$
defines the terminal cost. As an example, a running cost function can
be a quadratic cost $\ell(x,u) = \norm{x}_R +\norm{u}_Q$ for some
positive semi-definite matrices $Q$ and $R$, and a terminal cost can
be $g(x,u) = \norm{x-x_f}_R$ where $x_f$ is a goal
state. 


We denote the set of feedback policies to be $\Pi$.  For infinite
horizon optimal control, the optimal policy is independent of time and
a feedback controller suffices to be a minimizing argument of
\eqref{eq:originproblem} (see Ref. \cite{bertsekas1995dynamic}). For
finite-horizon optimal control, the optimal policy is
time-dependent. However, for simplicity, in this paper, we only
consider time-invariant feedback policies and assume the time horizon
$T$ is of sufficient length to ignore the time constraints.


\subsection{Preliminary: Model reference adaptive search}

\ac{mras} algorithm, introduced in \cite{hu2007model}, aims to solve
the following problem:
\[
z^\ast \in \arg \max_{z\in Z} H(z),\quad z\in \reals^n
\]
where $Z$ is the solution space and $H:\reals^n\rightarrow \reals$ is
a deterministic function that is bounded from below. It is assumed
that the optimization problem has a unique solution, i.e.,
$z^\ast \in Z$ and for all $z\ne z^\ast$, $H(z) < H(z^\ast)$.

The following regularity conditions need to be met for the
applicability of \ac{mras}.
\begin{assumption}
\label{assume1}
For any given constant $\xi < H(z^\ast)$, the set
$\{z \mid H(z) \ge \xi\} \cap Z$ has a strictly positive Lebesgue or
discrete measure.
\end{assumption}
This condition ensures  that any neighborhood of the optimal solution $z^\ast$ will have a positive probability to be sampled.
\begin{assumption}
\label{assume2}
  For any constant $\delta >0$, $\sup_{z\in A_\delta}H(z)< H(z^\ast)$,
  where $A_\delta := \{z\mid \norm{z-z^\ast}\ge \delta\}\cap X$, and
  we define the supremum over the empty set to be $-\infty$.
\end{assumption} 

\begin{itemize}
\item Selecte a sequence of reference distributions $\{g_k(\cdot)\}$
  with desired convergence properties. Specifically, the sequence
  $\{g_k(\cdot)\}$ will converge to a distribution that concentrates
  only on the optimal solution.
\item Selecte a parametrized family of distribution $f(\cdot, \theta)$
  over $X$ with parameter $\theta \in \Theta$.
\item Optimize the parameters $\{\theta_k\}$ iteratively by minimizing
  the following KL distance between $f(\cdot, \theta_k) $ and
  $g_k(\cdot)$.
\[
d(g_k, f(\cdot, \theta)):= \int_{Z}\ln\frac{g_k(z)}{f(z,\theta)}g_k(z)\nu(dz).
\]
where $\nu(\cdot)$ is the Lebesgue measure defined over $ Z$. The
sample distributions $\{f(\cdot,\theta_k)\}$ can be viewed as compact
approximations of the reference distributions and will converge to an
approximate optimal solution as $ \{g_k(\cdot)\}$ converges provided
certain properties of $\{g_k(\cdot)\}$ is retained in
$f(\cdot, \theta_k)$.
\end{itemize}
Note that the reference distribution $\{g_k(\cdot)\}$ is unknown
beforehand as the optimal solution is unknown. Thus, the \ac{mras}
algorithm employs the estimation of distribution algorithms
\cite{hauschild2011introduction} to estimate a reference distribution
that guides the search. To make the paper self-contained, we will
cover details of \ac{mras} in the development of the planning
algorithm.


\section{Approximate optimal motion planning using \ac{mras}}

In this section, we present an algorithm that uses \ac{mras} in a
distinguished way for approximate optimal feedback motion planning.


\subsection{Policy function approximation}

The \emph{policy function approximation} $\bar u: X\rightarrow U$ is a
weighted sum of basis functions,
\[
\bar u(x) = \sum_{i=1}^N w_i \phi_i(x)
\]
where $\phi_i : X\rightarrow \reals, i=1,\ldots, N$ are basis
functions, and the coefficients $w_i$ are the weight parameters,
$i=1,\ldots, N$. An example of basis function can be polynomial basis
$\phi=[1,x, x^2,x^3,\ldots, x^N]$ for one-dimensional system. A
commonly used class of basis functions is \ac{rbf}. It can be
constructed by determining a set of centers $c_i,\ldots, c_N \in X$,
and then constructing \ac{rbf} basis functions
$\phi_i = \exp (- \frac{\norm{x-c_i}^2}{2\sigma^2})$, for each center
$c_i$, where $\sigma$ is a pre-defined parameter.

In vector form, a policy function approximation is represented by
$\bar u = \langle w, \phi\rangle $ where vector
$\phi = \left[ \phi_1,\ldots \phi_N \right]^\intercal$ and
$w = [w_1,\ldots, w_N]^\intercal$. We let the domain of weight vector
be $W$ and denote it by
$\Pi_\phi = \{ \langle w, \phi \rangle \mid w \in W, \langle w, \phi
\rangle \in \Pi \}$
the set of all policies that can be generated by linear combinations
of pre-defined basis functions. In the following context, unless
specifically mentioned, the vector of basis functions is $\phi$.

Clearly, for any weight vector $w$,
$J(x_0, \langle w, \phi\rangle)\ge \min_{u\in \Pi}J(x_0, u)$. Thus, we
aim to solve $\min_{w} J(x_0, \langle w, \phi\rangle)$ so as to
minimize the error in the optimal cost introduced by policy function
approximation.

\begin{definition}[Approximate optimal feedback policy]
\label{def:optweight}
Given a basis vector $\phi$, a weight vector $w^\ast$ with respect to
$\phi$ is \emph{optimal} if and only if
$\langle w^\ast, \phi\rangle \in \Pi_\phi$ and for all $ w\in W$ such
that $ \langle w, \phi \rangle\in \Pi_\phi$,
\[
J(x_0, \langle w^\ast, \phi \rangle) \le J(x_0, \langle w,
\phi\rangle).
\]
The approximate optimal feedback policy is $\bar {u}^\ast =
\langle w^\ast, \phi \rangle$.
\end{definition}
By requiring
$ J(x_0, \langle w^\ast, \phi \rangle) \le J(x_0, \langle w,
\phi\rangle) $,
it can be shown that the optimal weight vector $w^\ast$ minimizes the
difference between the optimal cost achievable with policies in
$\Pi_\phi$ and the cost under the global optimal policy.

For clarity in notation, we denote
$J(x_0, \langle w, \phi\rangle)$ by $J(x_0; w)$ as $\phi$ is a fixed
basis vector throughout the development of the proposed method.


Clearly, if the actual optimal policy $u^\ast$ can be represented by a
linear combination of selected basis functions, then we obtain the
optimal policy by computing the optimal weight vector, i.e.,
$u^\ast = \langle w^\ast ,\phi\rangle$.

\paragraph*{Remark:} Here, we assume a feedback policy can be represented by
  $\langle w, \phi \rangle$ for some weight vector $w \in W$. In cases
  when the basis functions are continuous, a feedback policy must
  be a continuous function of the state. However, this requirement is hard
  to satisfy for many physical systems due to, for example, input
  saturation. In cases when a feasible controller is discontinuous, we
  can still use a continuous function to approximate, and then project
  the continuous function to the set $\Pi$ of applicable controllers.

  Using function approximation, we aim to solve the optimal feedback
  planning problem in \eqref{eq:originproblem} approximately by
  finding the optimal weight vector with respect to a pre-defined
  basis vector. The main algorithm is presented next.




\subsection{Integrating \ac{mras} in approximate optimal
  planning}
In this section, we present an adaptive search-based algorithm to
compute the approximate optimal feedback policy. The algorithm is
``near'' anytime, meaning that it returns a feasible solution after a
small number of samples. If more time is permitted, it will quickly
converge to the globally optimal solution that corresponds to the
approximate optimal feedback policy.  The algorithm is probabilistic
complete under regularity conditions of
\ac{mras}. 

We start by viewing the weight vector as a random variable $\rvw$
governed by a multivariate Gaussian distribution with a compact support $W$. The distribution is
parameterized by parameter $\theta =(\mu,\Sigma)$, where $\mu$ is a
$N$-dimensional mean vector and $\Sigma$ is the $N$ by $N$ covariance
matrix.  Recall $N$ is the number of basis functions.

The optimal weight vector $w^\ast$ can be represented as a
\emph{target distribution} $p_\goal$ as a Dirac's Delta, i.e.,
$p_\goal( w^\ast)=\infty $ and $p_\goal(w)=0$ for $w \ne w^\ast$.
Dirac's Delta is a special case of multivariate Gaussian distribution
with zero in the limit case of vanishing covariance. Thus, it is
ensured that the target distribution can be arbitrarily closely
approximated by multivariate Gaussian distribution by a realization of
parameter $\theta$.

Recall that the probability density of a multivariate Gaussian
distribution is defined by
\begin{align*}
&p(w; \theta) = \frac{1}{\sqrt{(2\pi)^N \abs{\Sigma}}}
\exp(-\frac{1}{2}(x-\mu)^\intercal\Sigma^{-1}(x-\mu)), \\
& \theta=
(\mu,\Sigma), \forall w\in W,
\end{align*}
where $N$ is the dimension of weight vector $w\in W$ and $\abs{\Sigma}$ is the
determinant of  $\Sigma$.


Now, we are ready to represent the main algorithm, called \ac{saop},
which includes the following steps.
\begin{enumerate}[1)]
\item \textbf{Initialization}: The initial distribution is selected to
  be $p(\cdot, \theta_0)$, for
  $\theta_0 =( \mu_0, \Sigma_0)\in \Theta$ which can generate a
  set of sample to achieve a good coverage of the sample space
  $W$. For example, $\mu_0 =\mathbf{ 0} \in \reals^N$ and
  $\Sigma_0 = \mathbf{ I}\in \reals^N $ which is an identity matrix.
  The following parameters are used in this algorithm: $\rho\in (0,1]$
  for specifying the quantile, the \emph{improvement parameter}
  $\varepsilon \in \reals^+$, a \emph{sample increment percentage} $\alpha$,  an initial sample size $N_1$, a
  \emph{smoothing coefficient} $\lambda \in (0,1]$. Let
  $k=1$. 
\item \textbf{Sampling-based policy evaluation}: At each iteration
  $k$, generate a set of $N_k$ samples $W_k \subseteq W $ from the
  current distribution $p(\cdot, \theta_k)$.  For each $w\in W_k $,
  using simulation we evaluate the cost $J(x_0 ; w)$ from the initial
  state $x_0$ and the feedback policy
  $u(x) =\langle w , \phi(x) \rangle $ with system model in
  \eqref{eq:sys}. The cost $J(x_0; w)$ is determined because the
  system is deterministic and has a unique solution.
\item \textbf{Policy improvement with Elite samples}: Next, the set
  $\{ J(x_0; w) \mid w  \in W_k \}$ is ordered from largest (worst)
  to smallest (best) among given samples:
\[
J_{k, (0)} \ge \ldots \ge J_{k, (N_k)}
\]
We denote $\kappa $ to be the estimated $(1-\rho)$-quantile of cost
$J(\cdot; w)$, i.e., $\kappa = J_{k, \lceil (1-\rho )N_k \rceil }$. 

The following cases are distinguished.
\begin{itemize}
\item If $k=1$, we introduce a threshold $ \gamma = \kappa$.
\item If $k \ne 1$, the following cases are further distinguished:
\begin{itemize}
\item $\kappa \le \gamma - \varepsilon$, i.e., the estimated
  $(1-\rho)$-quantile of cost has been reduced by the amount of
  $\varepsilon$ from the last iteration, then let $ \gamma=\kappa
  $. Let $N_{k+1}=N_k$ and continue to step 4).
\item Otherwise $\kappa > \gamma - \varepsilon$, we find the largest
  $ \rho' $, if it exists, such that the estimated
  $(1-\rho')$-quantile of cost
  $\kappa' = J_{k, \lceil (1-\rho' )N_k \rceil } $ satisfies
  $\kappa' \le \gamma -\varepsilon$. Then let $\gamma = \kappa'$ and
  also let $\rho = \rho'$.  Continue to step 4).  However, if no such
  $\rho'$ exists, then there is no update in the threshold $\gamma$
  but the sample size is increased to
  $N_{k+1} = \lceil (1+\alpha) N_k\rceil $. Let
  $\theta_{k+1} =\theta_k$, $k=k+1$, and continue to step 2).
\end{itemize}
\item \textbf{Parameter(Policy) update}: We update parameters
  $\theta_{k+1}$ for iteration $k+1$. First, we define a set
  $E = \{w \mid J(x_0; w) \le \gamma, w \in W_k, j=1,\ldots, k \}$ of
  \emph{elite samples}. Note that the parameter update in $\theta$ is
  to ensure a higher probability for elite samples. To achieve that,
  for each elite sample $w\in E $, we associated a weight such that a
  higher weight is associated with a weight vector with a lower cost
  and a lower probability in the current distribution. The next
  parameter $\theta_{k+1}$ is selected to maximize the weighted sum of
  probabilities of elite samples. To this end, we update the parameter
  as follows.
\begin{multline*}
  \theta_{k+1} ^\ast\\=\arg \max_{\theta \in \Theta}
  \mathbb{E}_{\theta_k}\left [\frac{S(J(x_0,
      w))^k}{p(w,\theta_k)}I_{J(x_0, w)\le \gamma}\ln
    p(w,\theta)\right]
\end{multline*}
where $\mathbb{E}_{\theta}(\nu)$ is the expected value of a random
variable $\nu$ given distribution $p(\cdot, \theta)$,
$S: \reals\rightarrow \reals^+$ is a strictly decreasing and positive
function \footnote{Possible choices can be $S(x) = \exp(-x)$ or
  $S(x)=\frac{1}{x}$ if $x$ is strictly positive.}.
$S(J(x_0; w))^k/p(w,\theta_k)$ is the weight for parameter $w$.
\end{itemize}
\end{enumerate}

\begin{assumption}
\label{assume3}
The optimal parameter $\theta^\ast$ is the interior point of $\Theta$
for all $k$.
\end{assumption} 

\begin{lemma}[based on Theorem~1 \cite{hu2007model}]
\label{lma}
  Assuming ~\ref{assume1},\ref{assume2}, and \ref{assume3} and the compactness of
  $W$, with probability one, 
\[
\lim_{k\rightarrow \infty}\mu_k = w^\ast, \text{ and }
\lim_{k\rightarrow \infty}\Sigma_k = 0_{N\times N}.
\]
where $w^\ast$ is the optimal weight vector and $0_{N \times N}$ is an
$N$-by-$N$ zero matrix.
\end{lemma}
Note that since $\Sigma_k$ converges in the limit a zero matrix, the stopping criterion is justified.


Building on the convergence result of \ac{mras}, the proposed
sampling-based planner ensures a convergence to a Dirac Delta function
concentrating on the optimum. In practice, the parameter update is
performed using the expectation---maximization (EM) algorithm.

\textbf{EM-based parameter update/policy improvement} Since our choice
of probability distribution is the multivariate Gaussian, the
parameter $\theta^\ast_{k+1} =(\vmu,\Sigma)$ is computed as follows
\begin{multline}
\label{eq:updatemean}
  \vmu = \frac{\bbE_{\theta_k} [ S(J(x_0,w))^k/p(w,\theta_k)]I_{w\in E} w} {\bbE_{\theta_k} [S(J(x_0,w))^k/p(w,\theta_k)]I_{w\in E}} \\
\approx \frac{\sum_{w\in W_k}[
    S(J(x_0,w))^k/p(w,\theta_k)]I_{w\in E} w}{ \sum_{w\in
      W_k} [S(J(x_0,w))^k/p(w,\theta_k)]I_{w\in E}},
\end{multline}
and 
\begin{multline}
\label{eq:updatecov}
  \Sigma= \frac{\bbE_{\theta_k} [S(J(x_0,w))^k/p(w,\theta_k)]I_{w\in
      E} (w - \vmu)(w - \vmu)^\intercal}{\bbE_{\theta_k}
    [S(J(x_0,w))^k/p(w,\theta_k)]I_{w\in E}}\\ \approx \frac{
    \sum_{w\in W_k} [S(J(x_0,w))^k/p(w,\theta_k)]I_{w\in E} (w -
    \vmu)(w - \vmu)^\intercal}{ \sum_{w\in
      W_k}[S(J(x_0,w))^k/p(w,\theta_k)]I_{w\in E}},
\end{multline}
where we approximate $\bbE_{\theta_k}(h(\rvw)) $ with its estimate
$ \frac{1}{N_k} \sum_{w\in W_k}h(w)$ for
$\rvw \sim p(\cdot, \theta_k)$ and the fraction $\frac{1}{N_k}$ was
canceled as the term is shared by the numerator and the denominator.


 \textbf{Smoothing}: Due to limited sample size, a greedy
  maximization for parameter update can be premature if too few
  samples are used. To ensure the convergence to the \emph{global
    optimal solution}, a \emph{smoothing} update is needed. To this
  end, we select the parameter for the next iteration to be
\[ \theta_{k+1} \leftarrow \lambda \theta_k +(1-\lambda )
\theta^\ast_{k+1}.
\]
where $\lambda \in [0,1)$ is the smoothing parameter.

Let $k=k+1$. We check if the iteration can be terminated based on a
given stopping criterion. If the stopping criterion is met, then we
output the latest $\theta_{k}$. Otherwise, we continue to update of
$\theta$ by moving to step 2).

\textbf{Stopping criterion} Given the probability distribution will
converge to a degenerated one that concentrates on the optimal weight
vector. We stop the iteration if the covariance matrix $\Sigma_k$
becomes near-singular given the convergence condition in Lemma~\ref{lma}.

To conclude, the proposed algorithm using \ac{mras} is probabilistic
complete and converges to the global optimal solution. If the
assumptions are not met, the algorithm converges to a local optimum.


\subsection{Robust control using trajectory verification in sampling}
Being able to directly search within continuous control policy space,
one major advantage is that one can enforce stability condition such
that the search is restricted to stable and robust policy space. In this subsection, we consider contraction theory to
compute conditions that need to be satisfied by weight vectors to
ensure stability and robustness under bounded disturbances.
\begin{definition}\cite{lohmiller1998contraction}
  Given the system equation for the closed-loop system
  $\dot x =f(x,t)$, a region of the state space is called a
  \emph{contraction region} if the Jacobian
  $\frac{\partial f}{\partial x}$ is uniformly negative definite in
  that region, that is,
\[
\exists \beta >0, \forall x, \forall t>0, \frac{1}{2}\left(
  \frac{\partial f}{\partial x}M+ \dot M+ M \frac{\partial f}{\partial x}^\intercal
\right) \preceq
 - \beta M.
\]
where $M(t)$ is a positive definite matrix for all $t\ge 0$.
\end{definition}

\begin{theorem}\cite{lohmiller1998contraction}
\label{thm:contraction}
  Given the system model $\dot x = f(x, t)$, any trajectory, which
  starts in a ball of constant radius  with respect to the matrix $M$, centered about a given
  trajectory and contained at all times in a contraction region with respect to the matrix $M$,
  remains in that ball and converges exponentially to this trajectory.
  Furthermore, global exponential convergence to the given trajectory
  is guaranteed if the whole state space is a contraction region.
\end{theorem}
Theorem~\ref{thm:contraction} provides a necessary and sufficient
condition for exponential convergence of an autonomous system. Under
bounded disturbances, the key idea is to incorporate a contraction
analysis in the planning algorithm such that it searches for a weight vector $w$
that is not only optimal in the nominal system but also ensures that
the closed-loop actual system under the controller
$u = w^\intercal \phi$ has contraction dynamics within a tube around
the nominal trajectory. Using a similar proof in \cite{Liu2014}, we
can show that for systems with contracting dynamics, the actual
trajectory under bounded disturbances will be ultimately uniformly
bounded along the nominal trajectory.

\begin{lemma}
  Consider a closed-loop system $\dot x =f(x)+\omega(t) $ where
  $\omega(t)$ is a disturbance with
  $\max_{t } \norm{\omega(t) }\le \rho_{\max} $, let a state
  trajectory $x(t)$ be in the contraction region $X_\ell$ at all time
  $t\ge t_0$, then for any time $t\ge t_0$, the deviation between
  $x(t)$ and the nominal trajectory $\bar x(t)$, whose dynamic model is
  given by $\dot {\bar x} = f(\bar x)$, satisfies
\[
\norm{x(t) - \bar x(t)}_M^2 \le \frac{2\ell \rho_{\max}}{\beta}(1- e^{\beta t}))
\] In other words, the error is uniformly ultimately bounded with the
ultimate bound $ \frac{2\ell \rho_{\max}}{\beta}$.
\end{lemma}
\begin{proof}
  Let's pick the Lyapunov function \[
V= (x-\bar x)^T M(x- \bar x),\] whose
  time derivative is
\begin{align*}
\dot V&  = (x -\bar x)^T M (f(x)+\omega -  f(\bar x))  \\
& + (f(x)+\omega - f(\bar
x))^T M(x- \bar x)  \\
& = (x-\bar x)^T M(f(x)-f(\bar x))  +2 (x-\bar x)^TM \omega \\
  & (M \mbox{ is symmetric})\\
& = (x - \bar x)^T ( \frac{\partial f}{\partial x}^T\mid_{\tilde x} M+M
  \frac{\partial f}{\partial x} \mid_{\tilde x} ) (x -\bar x) \\
& +2 (x-\bar x)^TM \omega,
\end{align*}
where the following property is used: $f(x)- f(\bar x) =
\frac{\partial f}{\partial x}^T \mid_{\tilde x} (x-\bar x)$ for some
$\tilde x \in [\bar x, x]$ if $\bar x \preccurlyeq x$ or $\tilde x \in
[x,\bar x]$ otherwise.

Since the trajectories stays within the contraction region, the following
condition holds
$\frac{\partial f}{\partial x}^T\mid_{\tilde x} M+M \frac{\partial
  f}{\partial x} \mid_{\tilde x} \le - 2\beta M$, and  we have
\[
\dot V \le -(x-\bar x)^T \beta M (x-\bar x)+2(x-\bar x)^T M\omega.
\]
Meanwhile, $\norm{x(t) - \bar x(t) }_M \le \ell$ as the trajectory $x(t)$
stays within the region of contraction, and also 
\[
 M (x- \bar x) \le \sqrt{(x-\bar x)^T M M (x-\bar x)}  =
 \sqrt{\norm{x-\bar x}_M}
\]

we conclude that as $\omega \le \rho_{\max}$,
\begin{multline*}
\dot V \le  -(x-\bar x)^T (\beta M) (x-\bar x) + 2\omega \sqrt{
  \norm{x - \bar x}_M}\\ \le  -(x-\bar x)^T (\beta M) (x-\bar x) + 2
\rho_{\max} \ell,
\end{multline*}

Since $\dot V \le -\beta V + 2\rho_{\max} \ell$ and under the condition that
$x(0 )= \bar x(0)$, we obtain 
$V(t) \le \frac{2\ell}{\beta} \rho_{\max}(1- e^{-\beta t})$, and
therefore 
\[
\norm{x-\bar x}_M^2=   \frac{2\ell}{\beta} \rho_{\max}(1- e^{-\beta t})
\]

\end{proof}
Thus, to search for the optimal and robust policies, we modify the
algorithm by introducing the following step.

\noindent \textbf{Contraction verification step:} Suppose the
closed-loop system is subject to bounded disturbances, the objective
is to ensure the trajectory is contracting within the time-varying
tube $\{x \mid \norm{x-\bar x} \le \ell\}$, for all $t$, where
$\bar x$ is the nominal state trajectory. The following condition
translates the contraction condition into verifiable condition for a
closed-loop system: Choose positive constants $\beta $, a positive
definite symmetric and constant matrix
$M = [m_{ij}]_{i=1,\ldots, n, j=1,\ldots, n}$, and verify whether, at
each time step along the nominal trajectory $ \bar x(t)$ in the
closed-loop system under control $ u(t)=w^\intercal \phi( x(t))$, the
following condition holds.
\begin{align}
\label{eq:robust}
\begin{split}
  \max_{x : \norm{x- \bar x}_M \le \ell } g_{ij}(x) \le -
    \beta m_{ij}, \quad \forall i=1,\ldots, n,\;\forall j=1,\ldots, n
\end{split}
\end{align}
where $g_{ij}$ is the $(i,j)$th component in the matrix
$\frac{\partial f}{\partial x}^\intercal M+M \frac{\partial
  f}{\partial x} $.
We verify this condition numerically at discrete time steps instead of
continous time.  Further, if the function $g_{ij}(x)$ is
semi-continuous, according to the Extreme Value Theorem, this
condition can be verified by evaluating $g_{ij}(x)$ at all critical
points where $\frac{dg_{ij} (x)}{dx}=0$ and the boundary of the set
$\{x\mid \norm{x-\bar x}_M \le \ell\}$.

The modification to the planning algorithm is made in Step 3), if a controller
$u = \langle w, \phi \rangle$ of elite sample $w$ does not meet the
condition, then $w$ is rejected from the set of elite
samples. Alternatively, one can do so implicitly by associating $w$
with a very large cost. However, since the condition is sufficient but
not necessary as we have the matrix $M$, constant $\beta$ and $\ell$
pre-fixed and $M$ is chosen to be a constant matrix, the obtained
robust controller may not necessary be optimal among all robust
controllers in $\Pi_\phi$. A topic for future work is to extend joint
planning and control policies with respect to adaptive bound $\beta$,
$\ell$, and a uniformly positive definite and time-varying matrix
$M(x,t)$.

\section{Examples}
In this section, we use two examples to illustrate the correctness and
efficiency of the proposed method. The simulation experiments are
implemented in MATLAB on a desktop with Intel Xeon E5 CPU and 16 GB of
RAM.

\subsection{Feedback gain search for linear systems}
To illustrate the correctness and sampling efficiency in the planning
algorithm, we consider an optimal control of \ac{lti} systems with
non-quadratic cost. For this class of optimal control problems, since
there is no admissible heuristic, one cannot use any planning
algorithm facilitated by the usage of a heuristic function. Moreover,
the optimal controller is nonlinear given the non-quadratic cost.

Consider a \ac{lti}
 system 
\[
\dot x = Ax + Bu
\]
where $A = \begin{bmatrix} -1 & 1 \\ 0& 0
\end{bmatrix}$ and $B=\begin{bmatrix}0\\1
\end{bmatrix}$
with $x\in X =\reals^2$ and $u\in \reals^1$. The initial state is
$x_0 = [5,5]$. 

The cost functional is 
$J(x_0,u) =\int_0^T( \norm{x}^2+ \norm{u}^2+ 0.5 \norm{x}^4+ 0.8\norm{x}^6 )dt
+ \norm{x(T)}^2. $

For a non-quadratic cost functional, the optimal controller is no
longer linear and cannot be computed by LQR unless the running cost
can be written in the sum-of-square form. Thus, we consider an
approximate feedback controller with basis vector $\phi =
[x_1,x_2,x_1^2,x_2^2, x_1^3,x_2^3]^\intercal$. Suppose the magnitude of external disturbance
is bounded by $\rho_{\max}= 0.5$.

The following parameters are used in stability verification:
$\beta = 2$, at any time $t$, for all $x$ such that
$\norm{x(t)- \bar x(t)} \le \ell$, the controller ensures
$\norm{x(t) - \bar x(t)} \le \frac{2\ell \rho_{\max}}{\beta}
(1-e^{\beta t})$
because $2\frac{\ell \rho_{\max}} {\beta} = 0.5 \ell \le \ell$. With
this choice for stability analysis,  the constraint 
\begin{align*}
\frac{\partial f}{\partial x} + \frac{\partial
    f}{\partial x}^\intercal & = \begin{bmatrix}-2 & 3 w(5) x_1^2 + 2 w(3)x_1
    + w(1 ) + 1\\
\mbox{Sym.} & 6 w(6)x_2^2 + 4 w(4) x_2 + 2w(2) \end{bmatrix}\\
& \le \begin{bmatrix}-2 &0 \\
0& -2\\
\end{bmatrix}
\end{align*}
In this case, if we select $w(6), w(4), w(5), w(3) $ nonpositive, 
$w(1)\le -1$ and $w(2)\le -1$, then closed-loop system, which is a
nonlinear polynomial system, will become globally contracting.

\begin{figure}[h]
\centering
\begin{subfigure}[b]{0.32\textwidth}
\centering
\includegraphics[width=\textwidth]{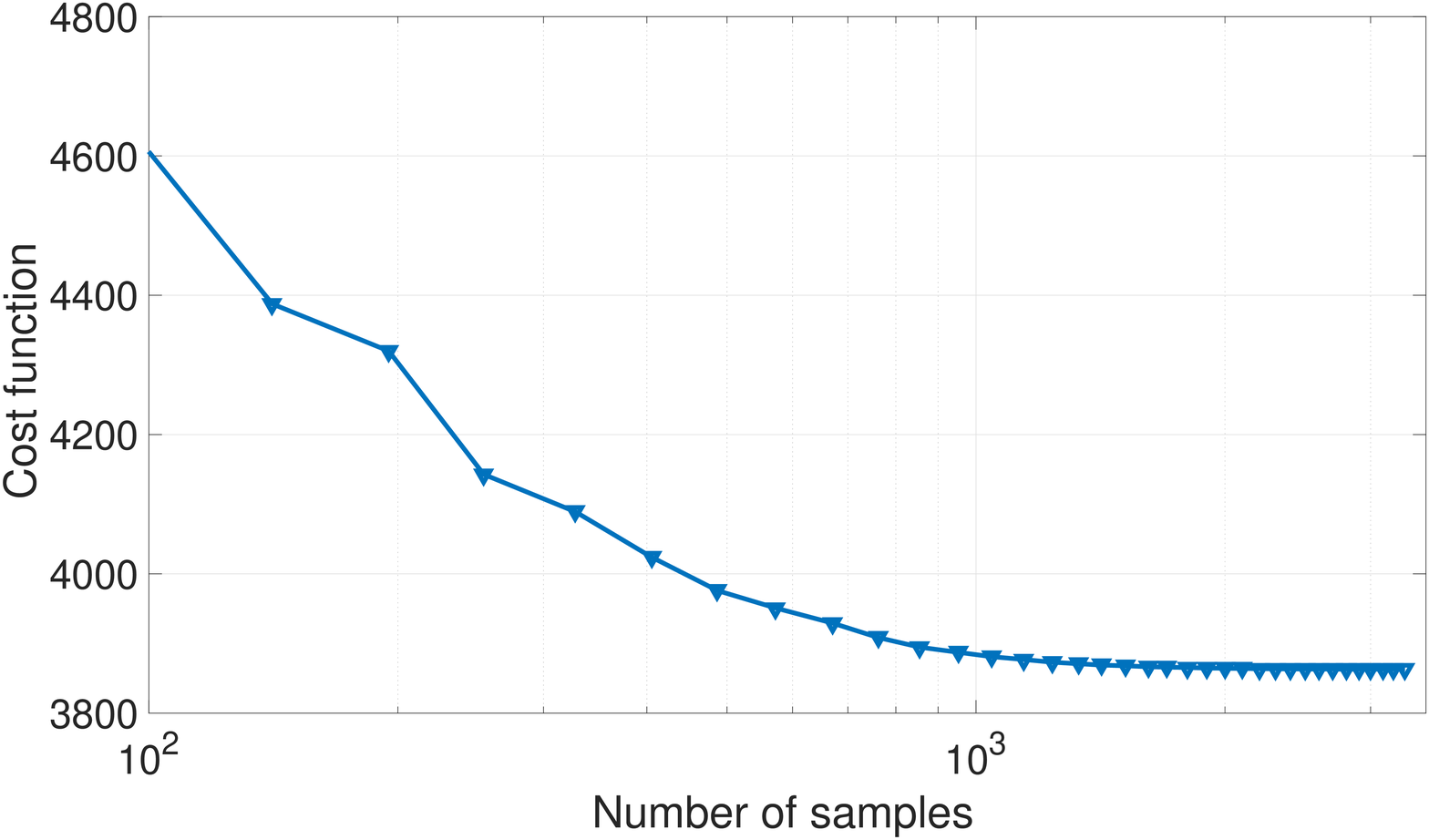}
\caption{\label{fig:cost}}
\end{subfigure}\hspace{-3ex}
\begin{subfigure}[b]{0.33\textwidth}
\centering
\includegraphics[width=\textwidth]{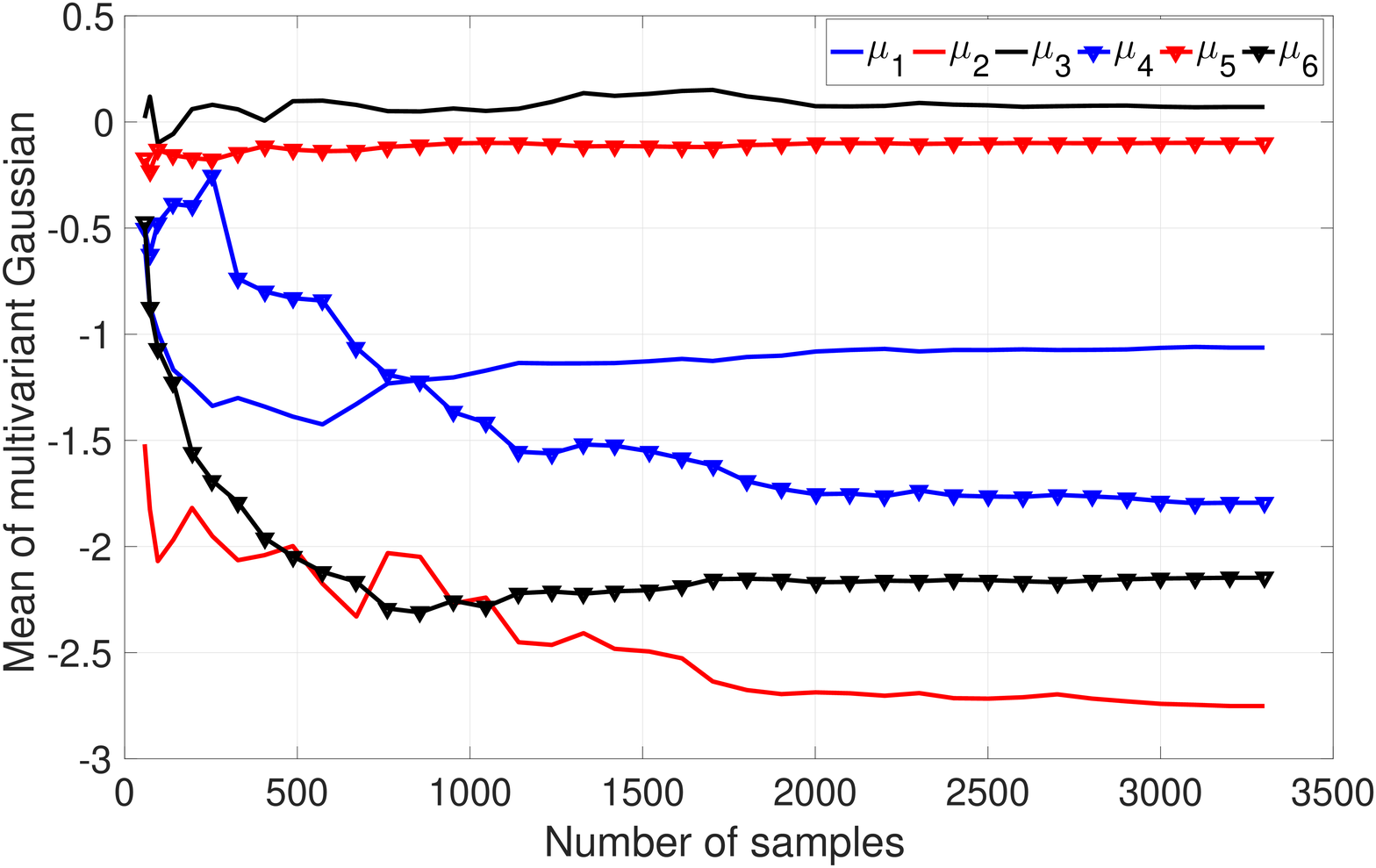}
\caption{\label{fig:mean}}
\end{subfigure}\hspace{-3ex}
\begin{subfigure}[b]{0.33\textwidth}
\centering
\includegraphics[width= \textwidth]{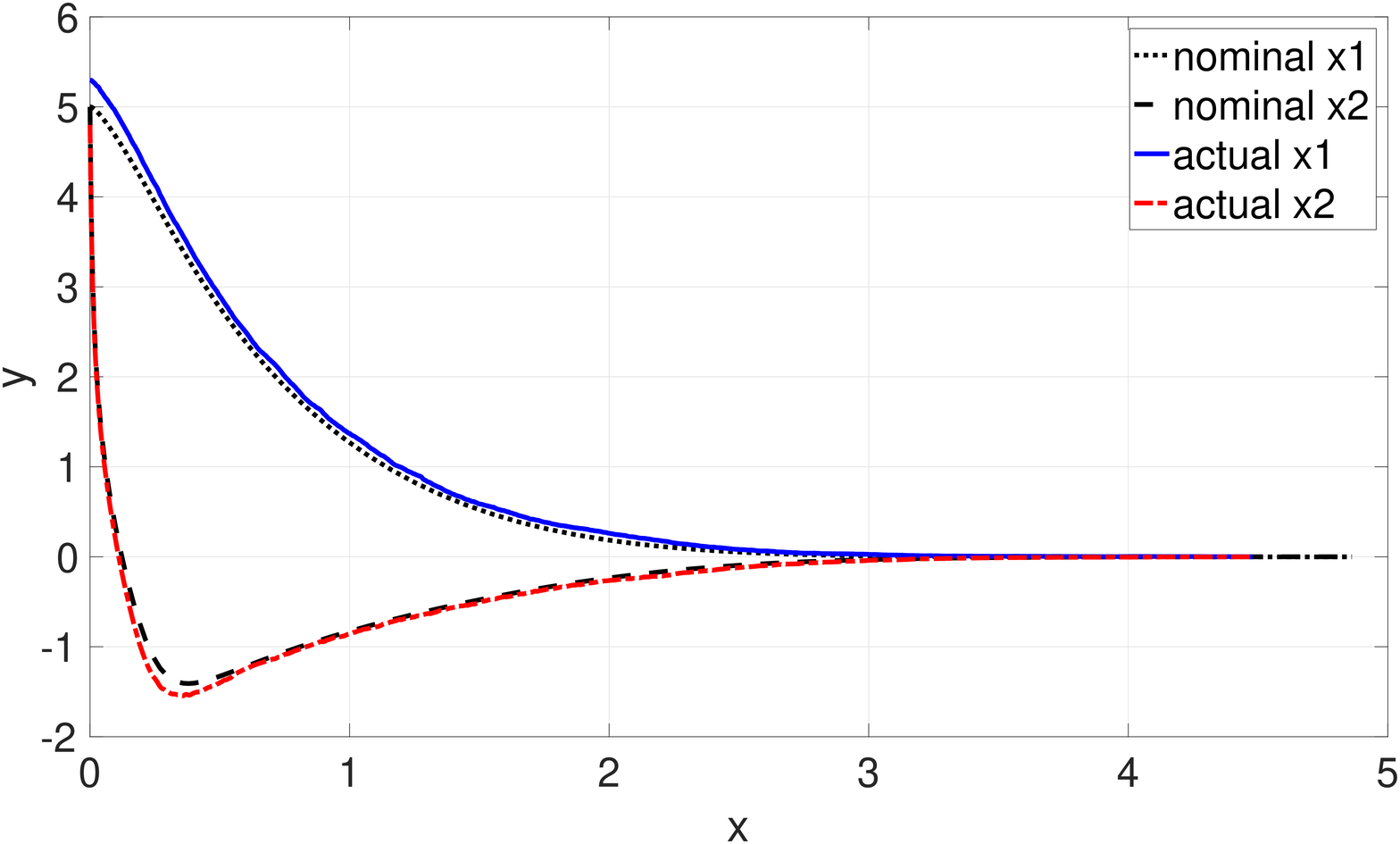}
\caption{\label{fig:statelinnoquad}}
\end{subfigure}
\caption{Convergence of the \ac{saop} algorithm on the \ac{lti} system
  with a nonquadratic cost functional. (a) The mean of multivariate Gaussian as weight vector
  over iterations. (b) The state trajectory of the closed-loop system
  under bounded disturbance $\rho_{\max}= 0.5$ under feedback
  controller computed with \ac{saop}.}
\label{fig:linquad}
\end{figure}
Figures~\ref{fig:cost} and \ref{fig:mean} show the convergence result
with \ac{saop} in one simulation in terms of cost and the mean of the
multivariant Gaussian over iterations. The following parameters are
used: Initial sample size $N_1 =50$, improvement parameter
$\epsilon =0.1$, quantile percentage $\rho=0.1$, smoothing parameter
$\lambda= 0.5$, sample increment parameter $\alpha = 0.1$. 

The algorithm converges after $38$ iterations with $3301$ samples to
the mean
$\bar w^\ast = [ -1.0629 \; -2.7517 \; 0\; -1.7939 \; -0.0987\;
-2.1474 ]^\intercal $
and the covariance matrix with a norm $\num{3.3401e-04}$.  Each
iteration took less than 10 seconds. The
approximate optimal cost under feedback controller
$u =\langle \bar w^\ast , \phi\rangle$ is $ 3863.3$.
Figure~\ref{fig:statelinnoquad} shows the state trajectory for the
closed-loop system with bounded disturbances.  With 25 independent
runs of \ac{saop}, the mean of $J(x_0; \bar w_i^\ast), i=1,\ldots, 25$
is $3903.3$ and the standard deviation is $ 104.1683 $, $2.6\% $ of
the approximate optimal cost.

Note, if we only use linear feedback $u= Kx$, the optimal cost is
$\num{ 1.0943e+04} $, which is about three times the optimal cost that
can be achieved with a nonlinear controller.

\subsection{Example: approximate optimal planning of a Dubins car} 

Consider a Dubins car dynamics
\[\dot x  = u\cos \theta,\quad \dot y  = u \sin \theta \; \quad \dot \theta = v\]
where $\vec{x}= (x,y,\theta)\in \reals^2\times \mathbb{S}^1$ being the
state (coordinates and turning angle with respect to $x$-axis) and $u$
and $v$ are control variables including linear and angular
velocities. The system is kinematically constrained by its positive
minimum turning radius $r$ which implies the following bound
$\abs{v} \le \frac{1}{r} $. Without loss of generality, we assume
$\abs{v}\le 5$ and $\abs{u}\le 10$ are the input constraints. The
control objective is to reach the goal $x_f=20, y_f=20$ while avoiding
static obstacles. The cost function
$J= \int_{0}^T \ell(x,u)dt + g(x,u)$ where $T=100$, the running cost
is $\ell(x,u) = 0.1\times (\norm{x} + \norm{u})$, and the terminal
cost is $g(x(T),u(T)) =1000\times \norm{(x(T),y(T))-(x_f,y_f)}$. The
initial state is $\vec{x}_0 = \origin$. In simulation, we consider the
robot reaches the target if $\norm{(x,y)'-(x_f,y_f)} \le \varepsilon$
for $\varepsilon \in [0, 1]$. In simulation, $\varepsilon =0.5$.

We select \ac{rbf} as basis functions and define
$\phi_{rbf} = [\phi_1,\ldots, \phi_N]^\intercal$ for $N$ center
points. In the experiment, the center points are
includes \begin{inparaenum}[1)]\item uniform grids in $x-y$
  coordinates with step sizes $\delta x = 5$, $\delta y =5$; and \item
  vertices of the obstacle. \end{inparaenum} We also include linear
basis functions $\phi_{linear} = [(x-x_f), (y-y_f),\theta]$.  The
basis vector is
$\phi = [\phi_{rbf}^\intercal, \phi_{linear}^\intercal]^\intercal$. We
consider a bounded domain $-5\le x \le 30$ and $-5\le y\le 30$ and
$\theta \in [0, 2\pi]$ and thus the total number of basis functions is
$80$.  The control input $\vec{u}= [u,v]^\intercal$ where
$u = w_u^\intercal \phi$ and $v =w_v^\intercal \phi $. The total
number of weight parameters is twice the number of bases and in this
case $160$.

\begin{figure}[t!]
\begin{subfigure}[t]{0.45\textwidth}
\centering
\includegraphics[width=\textwidth]{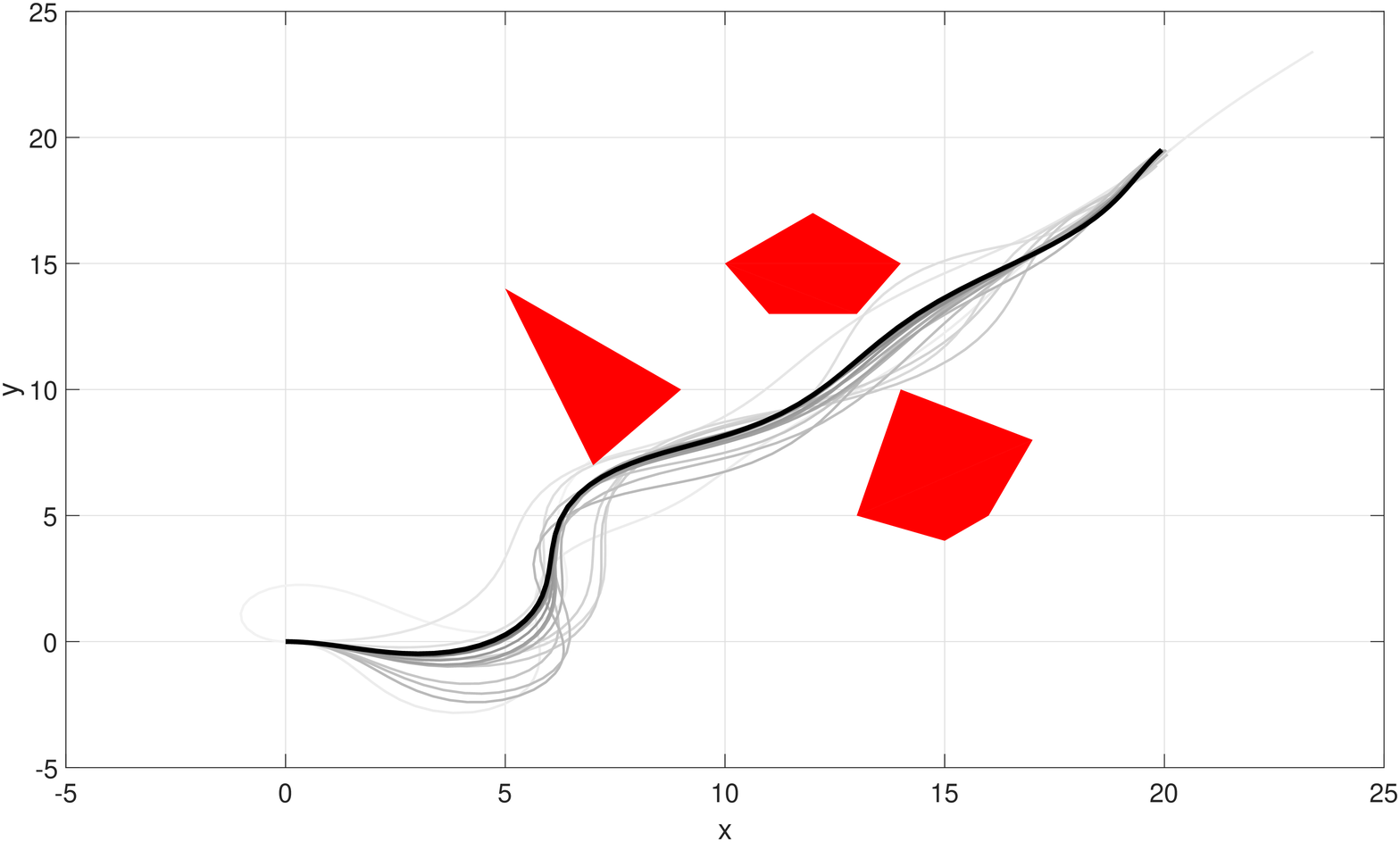}
\caption{\label{fig:dubinstraj}
}
\end{subfigure}
\begin{subfigure}[t]{0.45\textwidth}
\centering
\includegraphics[width= \textwidth]{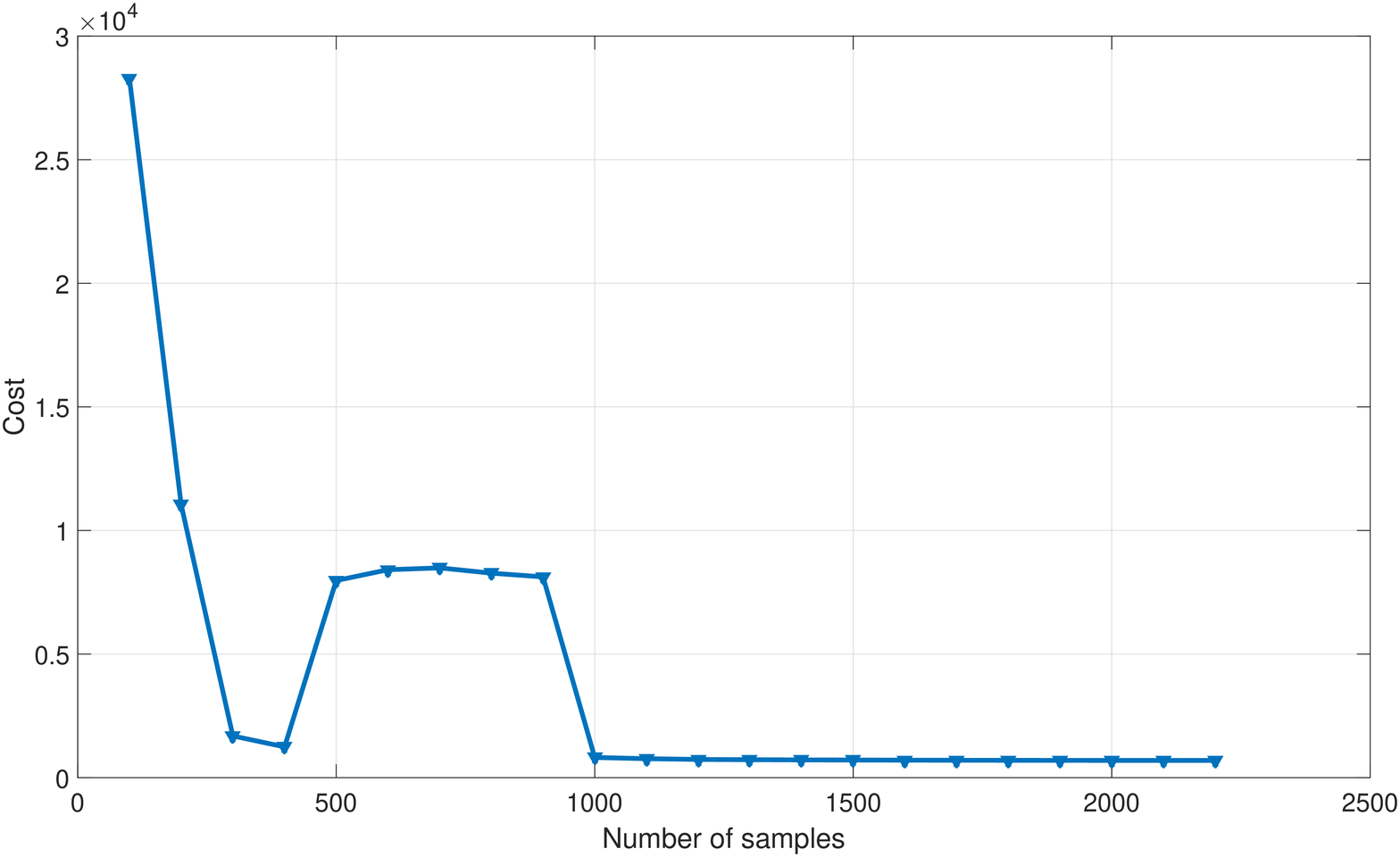}
\caption{\label{fig:dubinscost}}
\end{subfigure}
\caption{Convergence of the the planning algorithm algorithm on the Dubins
  car. (a) The planned trajectory under feedback policy
  $\langle \mu, \phi \rangle$ computed using the mean of multivariate
  Gaussian over iterations (from the lightest to the darkest). (b) The convergence of the covariance
  matrix. (c) The total cost evaluated at the mean of the multivariate
  Gaussian over iterations. }
\label{fig:dubinscar_conv}
\vspace{-3ex}
\end{figure}


The following parameters are used: Initial sample size $N_1= 100$,
improvement parameter $\epsilon = 0.1 $, smoothing parameter
$\lambda= 0.5$, sample increment percentage $\alpha=0.1$, and
$\rho = 0.1$. In Fig.~\ref{fig:dubinstraj} we show the trajectory
computed using the estimated mean of multivariate Gaussian
distribution over iterations, from the lightest ($1$-th iteration) to
the darkest (the last iteration when stopping criterion is met). The
optimal trajectory is the darkest line. In Fig.~\ref{fig:dubinscost}
we show the cost computed using the mean of multivariate Gaussian over
iterations.  \ac{saop} converges after 22 iterations with $2200$
samples and the optimal cost is $697.29$. Each iteration took about
$20$ to $ 30$ seconds.  However, it generates a collision-free path
only after $5$ iterations. Due to input saturation, the algorithm is
only ensured to converge to a local optimum. However, in 24
independent runs, all runs converges to a local optimum closer to the
global one, as shown in the histogram in Fig.~\ref{fig:dubins_histo}.
Our current work is to implement trajectory-based contraction analysis
using time-varying matrices $M(x,t)$ and adaptive bound $\beta$, which
are needed for nonlinear Dubins car dynamics.

\begin{figure}
\centering
\includegraphics[width= 0.5 \textwidth]{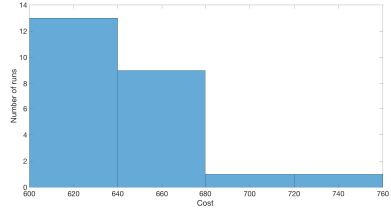}
\caption{The frequency distribution of the optimal costs with 24
  independent runs.}
\label{fig:dubins_histo}
\end{figure}

\section{Conclusion}
\label{sec:conclude}
In this paper, an importance sampling-based approximate optimal
planning and control method is developed. In the control-theoretic
formulation of optimal motion planning, the planning algorithm
performs direct policy computation using simulation-based adaptive
search for an optimal weight vector corresponding to an approximate
optimal feedback policy. Each iteration of the algorithm runs time
linear in the number of samples and in the time horizon for simulated
runs. However, it is hard to quantify the number of iterations
required for \ac{mras} to converge. One future work is to consider
incorporate multiple-distribution importance sampling to achieve
faster and better convergence results.  Based on contraction analysis
of the closed-loop system, we show that by modifying the
sampling-based policy evaluation step in the algorithm, the proposed
planning algorithm can be used for joint planning and robust control
for a class of nonlinear systems under bounded disturbances. In future
extension of this work, we are interested in extending this algorithm
for stochastic optimal control.

%
\bibliographystyle{IEEEtran}



\end{document}